\pdfoutput=1

\newif\ifdraft\draftfalse 
\newif\ifanon\anontrue    
\newif\iffull\fullfalse   
\newif\ifsooner\soonerfalse 
\newif\ifcamera\camerafalse 
\makeatletter \@input{texdirectives.tex} \makeatother

\documentclass[10pt, conference, compsocconf, letterpaper, times]{IEEEtran}
\IEEEoverridecommandlockouts
\usepackage{latexsym}
\usepackage{ifthen}
\usepackage{amsmath,amssymb,amsfonts}
\usepackage{algorithmic}
\usepackage{graphicx}
\usepackage{textcomp}
\usepackage{xcolor}
\usepackage{colortbl}
\def\BibTeX{{\rm B\kern-.05em{\sc i\kern-.025em b}\kern-.08em
    T\kern-.1667em\lower.7ex\hbox{E}\kern-.125emX}}
\usepackage{tikz}
\usetikzlibrary{arrows, calc, fit, positioning, shapes}
\usepackage{amsthm}
\usepackage{thmtools}
\declaretheorem{theorem}

\declaretheorem[name=Proposition, style=theorem]{proposition}

\declaretheorem[name=Definition, style=definition]{definition}

\usepackage{mathpartir}
\usepackage[numbers,sort]{natbib}
\newcommand{\squeezeup}{\vspace{-2.5mm}}

\usepackage[utf8]{inputenc}
\usepackage{newunicodechar}

\ifdraft
\usepackage{todonotes}
\presetkeys{todonotes}{inline}{}
\else
\newcommand{\todo}[1]{}
\fi

\usepackage[frozencache]{minted}
\newcommand{\hsi}[1]{\mintinline{text}{#1}}
\newcommand{\ai}[1]{{\normalfont\mintinline{text}{#1}}}
\newcommand\aim[1]{{\normalfont\text{\texttt{#1}}}} %
\setminted{fontsize=\small,baselinestretch=1}
\usepackage{etoolbox}
\BeforeBeginEnvironment{minted}{\smallskip}
\AfterEndEnvironment{minted}{\smallskip}
\newenvironment{haskell}
    {\VerbatimEnvironment
      \begin{minted}{text}}
    {
    \end{minted}
    }
\newenvironment{agda}
    {\VerbatimEnvironment
      \normalfont

      \begin{minted}[escapeinside=&&]{agda}}
    {
    \end{minted}
    }
\definecolor{keyword-yellow}{rgb}{0.8, 0.6, 0.3}
\newcommand{\kw}[1]{{\color{keyword-yellow} #1}}
\usepackage{mathtools}
\usepackage{stmaryrd}
\usepackage{multicol}
\usepackage{xspace}
\usepackage{bbold}
\usepackage{combelow}

\definecolor{dkblue}{rgb}{0,0.1,0.5}
\definecolor{dkgreen}{rgb}{0,0.4,0}
\definecolor{dkred}{rgb}{0.6,0,0}
\definecolor{dkpurple}{rgb}{0.7,0,1.0}
\definecolor{purple}{rgb}{0.9,0,1.0}
\definecolor{olive}{rgb}{0.4, 0.4, 0.0}
\definecolor{teal}{rgb}{0.0,0.4,0.4}
\definecolor{azure}{rgb}{0.0, 0.5, 1.0}
\definecolor{gray}{rgb}{0.5, 0.5, 0.5}
\definecolor{dkgray}{rgb}{0.3, 0.3, 0.3}

\usepackage[
    pdftex,%
    pdfpagelabels,%
    colorlinks,%
    linkcolor=dkblue,%
    citecolor=dkblue,%
    filecolor=dkblue,%
    urlcolor=dkblue%
]{hyperref}
\usepackage[capitalise,noabbrev]{cleveref}

\def\IEEEbibitemsep{3pt}

\DeclareMathAlphabet{\mathit}{\encodingdefault}{\familydefault}{m}{it}

\def\Snospace~{\S{}}

\def\Nnospace~{}

\newcommand{\comm}[3]{\ifdraft{{\color{#1}[#2: #3]}}\fi}
\newcommand{\ma}[1]{\comm{dkpurple}{MA}{#1}}
\newcommand{\ch}[1]{\comm{teal}{CH}{#1}}

\newcommand*{\EG}{e.g.,\xspace}
\newcommand*{\IE}{i.e.,\xspace}
\newcommand*{\ETAL}{et al.\xspace}

\newunicodechar{ᵢ}{\ensuremath{{}_i}}
\newunicodechar{ⱼ}{\ensuremath{{}_j}}
\newunicodechar{ₒ}{\ensuremath{{}_o}}
\newunicodechar{ⁱ}{\ensuremath{{}^i}}
\newunicodechar{ }{{\,}}
\newunicodechar{¬}{\ensuremath{\neg}} %
\newunicodechar{²}{^2}
\newunicodechar{³}{^3}
\newunicodechar{·}{\ensuremath{\cdot}}
\newunicodechar{¹}{^1}
\newunicodechar{×}{\ensuremath{\times}} %
\newunicodechar{÷}{\ensuremath{\div}} %
\newunicodechar{ʳ}{^r}
\newunicodechar{ᵣ}{\ensuremath{_r}}
\newunicodechar{ˡ}{^l}
\newunicodechar{Γ}{\ensuremath{\Gamma}}   %
\newunicodechar{Δ}{\ensuremath{\Delta}} %
\newunicodechar{Η}{\ensuremath{\textrm{H}}} %
\newunicodechar{Θ}{\ensuremath{\Theta}} %
\newunicodechar{Λ}{\ensuremath{\Lambda}} %
\newunicodechar{Ξ}{\ensuremath{\Xi}} %
\newunicodechar{Π}{\ensuremath{\Pi}}   %
\newunicodechar{Σ}{\ensuremath{\Sigma}} %
\newunicodechar{Φ}{\ensuremath{\Phi}} %
\newunicodechar{Ψ}{\ensuremath{\Psi}} %
\newunicodechar{Ω}{\ensuremath{\Omega}} %
\newunicodechar{α}{\ensuremath{\mathnormal{\alpha}}}
\newunicodechar{α}{\ensuremath{\mathnormal{\alpha}}} %
\newunicodechar{β}{\ensuremath{\beta}} %
\newunicodechar{γ}{\ensuremath{\mathnormal{\gamma}}} %
\newunicodechar{δ}{\ensuremath{\mathnormal{\delta}}} %
\newunicodechar{ε}{\ensuremath{\mathnormal{\varepsilon}}} %
\newunicodechar{ζ}{\ensuremath{\mathnormal{\zeta}}} %
\newunicodechar{η}{\ensuremath{\mathnormal{\eta}}} %
\newunicodechar{θ}{\ensuremath{\mathnormal{\theta}}} %
\newunicodechar{ι}{\ensuremath{\mathnormal{\iota}}} %
\newunicodechar{κ}{\ensuremath{\mathnormal{\kappa}}} %
\newunicodechar{λ}{\ensuremath{\mathnormal{\lambda}}} %
\newunicodechar{μ}{\ensuremath{\mathnormal{\mu}}} %
\newunicodechar{ν}{\ensuremath{\mathnormal{\mu}}} %
\newunicodechar{ξ}{\ensuremath{\mathnormal{\xi}}} %
\newunicodechar{π}{\ensuremath{\mathnormal{\pi}}}
\newunicodechar{π}{\ensuremath{\mathnormal{\pi}}} %
\newunicodechar{ρ}{\ensuremath{\mathnormal{\rho}}} %
\newunicodechar{σ}{\ensuremath{\mathnormal{\sigma}}} %
\newunicodechar{τ}{\ensuremath{\mathnormal{\tau}}} %
\newunicodechar{φ}{\ensuremath{\mathnormal{\varphi}}} %
\newunicodechar{χ}{\ensuremath{\mathnormal{\chi}}} %
\newunicodechar{ψ}{\ensuremath{\mathnormal{\psi}}} %
\newunicodechar{ω}{\ensuremath{\mathnormal{\omega}}} %
\newunicodechar{ϕ}{\ensuremath{\mathnormal{\phi}}} %
\newunicodechar{ϵ}{\ensuremath{\mathnormal{\epsilon}}} %
\newunicodechar{ᵏ}{^k}
\newunicodechar{ }{{\,}}
\newunicodechar{′}{\ensuremath{^\prime}}  %
\newunicodechar{″}{\ensuremath{^\second}} %
\newunicodechar{‴}{\ensuremath{^\third}}  %
\newunicodechar{⁵}{\ensuremath{^5}}
\newunicodechar{⁺}{\ensuremath{^+}} 
\newunicodechar{⁺}{^+}
\newunicodechar{⁻}{\ensuremath{^-}} 
\newunicodechar{ⁿ}{^n}
\newunicodechar{₀}{\ensuremath{_0}} %
\newunicodechar{₁}{\ensuremath{_1}}
\newunicodechar{₂}{\ensuremath{_2}}
\newunicodechar{ℓ}{\ensuremath{\ell}}
\newunicodechar{₊}{\ensuremath{_+}} 
\newunicodechar{₋}{\ensuremath{_-}} 
\newunicodechar{ℂ}{\ensuremath{\mathbb{C}}} %
\newunicodechar{ℕ}{\ensuremath{\mathbb{N}}} %
\newunicodechar{ℝ}{\ensuremath{\mathbb{R}}} %
\newunicodechar{ℤ}{\ensuremath{\mathbb{Z}}} %
\newunicodechar{ℳ}{\mathcal{M}}
\newunicodechar{⅋}{\ensuremath{\parr}} %
\newunicodechar{←}{\ensuremath{\leftarrow}} %
\newunicodechar{↑}{\ensuremath{\uparrow}} %
\newunicodechar{→}{\ensuremath{\rightarrow}} %
\newunicodechar{↔}{\ensuremath{\leftrightarrow}} %
\newunicodechar{↖}{\nwarrow} %
\newunicodechar{↗}{\nearrow} %
\newunicodechar{↝}{\ensuremath{\leadsto}}
\newunicodechar{↦}{\ensuremath{\mapsto}}
\newunicodechar{⇆}{\ensuremath{\leftrightarrows}} %
\newunicodechar{⇐}{\ensuremath{\Leftarrow}} %
\newunicodechar{⇒}{\ensuremath{\Rightarrow}} %
\newunicodechar{⇔}{\ensuremath{\Leftrightarrow}} %
\newunicodechar{∀}{\ensuremath{\forall}}   %
\newunicodechar{∃}{\ensuremath{\exists}} %
\newunicodechar{∅}{\ensuremath{\varnothing}} %
\newunicodechar{∈}{\ensuremath{\in}}
\newunicodechar{∉}{\ensuremath{\not\in}} %
\newunicodechar{∋}{\ensuremath{\ni}}  %
\newunicodechar{∎}{\ensuremath{\qed}}
\newunicodechar{∏}{\prod}
\newunicodechar{∑}{\sum}
\newunicodechar{∗}{\ensuremath{\ast}} %
\newunicodechar{∘}{\ensuremath{\circ}} %
\newunicodechar{∙}{\ensuremath{\bullet}} %
\newunicodechar{∙}{\ensuremath{\cdot}}
\newunicodechar{∞}{\ensuremath{\infty}} %
\newunicodechar{∣}{\ensuremath{\mid}} %
\newunicodechar{∧}{\ensuremath{\wedge}}%
\newunicodechar{∨}{\ensuremath{\vee}}%
\newunicodechar{∩}{\ensuremath{\cap}} %
\newunicodechar{∪}{\ensuremath{\cup}} %
\newunicodechar{∫}{\int}
\newunicodechar{∷}{::} %
\newunicodechar{∼}{\ensuremath{\sim}} %
\newunicodechar{≃}{\ensuremath{\simeq}} %
\newunicodechar{≅}{\ensuremath{\cong}} %
\newunicodechar{≈}{\ensuremath{\approx}} %
\newunicodechar{≜}{\ensuremath{\stackrel{\scriptscriptstyle {\triangle}}{=}}} 
\newunicodechar{≜}{\triangleq}
\newunicodechar{≝}{\ensuremath{\stackrel{\scriptscriptstyle {\text{def}}}{=}}}
\newunicodechar{≟}{\ensuremath{\stackrel {_\text{\textbf{?}}}{\text{\textbf{=}}\negthickspace\negthickspace\text{\textbf{=}}}}}
\newunicodechar{≠}{\ensuremath{\neq}}%
\newunicodechar{≡}{\ensuremath{\equiv}}%
\newunicodechar{≤}{\ensuremath{\le}} %
\newunicodechar{≥}{\ensuremath{\ge}} %
\newunicodechar{⊂}{\ensuremath{\subset}} %
\newunicodechar{⊃}{\ensuremath{\supset}} %
\newunicodechar{⊆}{\ensuremath{\subseteq}} %
\newunicodechar{⊇}{\ensuremath{\supseteq}} %
\newunicodechar{⊎}{\ensuremath{\uplus}} %
\newunicodechar{⊑}{\ensuremath{\sqsubseteq}} %
\newunicodechar{⊒}{\ensuremath{\sqsupseteq}} %
\newunicodechar{⊓}{\ensuremath{\sqcap}} %
\newunicodechar{⊔}{\ensuremath{\sqcup}} %
\newunicodechar{⊕}{\ensuremath{\oplus}} %
\newunicodechar{⊗}{\ensuremath{\otimes}} %
\newunicodechar{⊢}{\ensuremath{\vdash}} %
\newunicodechar{⊥}{\ensuremath{\bot}} 
\newunicodechar{⊧}{\models} %
\newunicodechar{⊨}{\models} %
\newunicodechar{⊩}{\Vdash}
\newunicodechar{⊸}{\ensuremath{\multimap}} %
\newunicodechar{⋁}{\ensuremath{\bigvee}}
\newunicodechar{⋃}{\ensuremath{\bigcup}} %
\newunicodechar{⋄}{\ensuremath{\diamond}} %
\newunicodechar{⋆}{\ensuremath{\star}} %
\newunicodechar{⋮}{\ensuremath{\vdots}} %
\newunicodechar{⋯}{\ensuremath{\cdots}} %
\newunicodechar{…}{\ensuremath{.\kern-0.10em.\kern-0.10em.}} %
\newunicodechar{─}{---}
\newunicodechar{□}{\ensuremath{\square}} %
\newunicodechar{▹}{\ensuremath{\rhd}} %
\newunicodechar{◃}{\triangleleft{}}
\newunicodechar{◇}{\ensuremath{\diamond}} %
\newunicodechar{◽}{\ensuremath{\square}}
\newunicodechar{★}{\ensuremath{\star}}   %
\newunicodechar{♭}{\ensuremath{\flat}} %
\newunicodechar{♯}{\ensuremath{\sharp}} %
\newunicodechar{✓}{\ensuremath{\checkmark}} %
\newunicodechar{⟂}{\ensuremath{^\bot}} 
\newunicodechar{⊤}{\ensuremath{\top}} 
\newunicodechar{⟦}{\ensuremath{\llbracket}}
\newunicodechar{⟧}{\ensuremath{\rrbracket}}
\newunicodechar{⦇}{\ensuremath{\llparenthesis}}
\newunicodechar{⦈}{\ensuremath{\rrparenthesis}}
\newunicodechar{⟨}{\ensuremath{\langle}} %
\newunicodechar{⟩}{\ensuremath{\rangle}} %
\newunicodechar{⟶}{{\longrightarrow}}
\newunicodechar{⟹}{\ensuremath{\Longrightarrow}} %
\newunicodechar{𝒟}{\ensuremath{\mathcal{D}}} %
\newunicodechar{𝒢}{\ensuremath{\mathcal{G}}}
\newunicodechar{𝒦}{\ensuremath{\mathcal{K}}} %
\newunicodechar{𝔸}{\ensuremath{\mathds{A}}} %
\newunicodechar{𝔹}{\ensuremath{\mathds{B}}} %
\newunicodechar{𝟙}{\ensuremath{\mathds{1}}} %
\newunicodechar{̷}{\ensuremath{\not}} %
\newunicodechar{̸}{\ensuremath{\not}} %
\newunicodechar{ᵇ}{\ensuremath{^b}} %

\begin{document}

\newcommand{\Nat}{\mathbb{N}}
\newcommand{\ME}{\text{ME}}
\newcommand{\MEF}{\text{MEF}}
\newcommand{\project}{\operatorname{\downarrow}}
\renewcommand{\L}{\mathcal{L}}
\newcommand{\lequiv}[1]{\approx_{#1}}
\newcommand{\Pow}[1]{\mathcal P(#1)}
\newcommand{\Bool}{\texttt{Bool}}
\newcommand{\rel}[1]{\llbracket #1 \rrbracket}
\newcommand{\sem}[1]{{#1}}
\newcommand{\many}[1]{\overline{#1}}
\renewcommand{\iff}{\Leftrightarrow}
\renewcommand{\implies}{\Rightarrow}
\newcommand{\subst}[3]{[{#2}/{#1}]{#3}}
\newcommand{\closed}[1]{\texttt{closed}({#1})}
\newcommand{\return}{\texttt{return}}
\newcommand{\Sec}{\texttt{Sec}}
\newcommand{\DCC}{\texttt{DCC}}
\newcommand{\up}{\texttt{up}}
\newcommand{\True}{\texttt{true}}
\newcommand{\False}{\texttt{false}}
\newcommand{\If}{\texttt{if}}
\newcommand{\bind}{\texttt{bind}}
\newcommand{\pair}{\texttt{pair}}
\newcommand{\either}{\texttt{either}}
\newcommand{\case}{\texttt{case}}
\newcommand{\canFlowTo}{\sqsubseteq}
\renewcommand{\S}{S}
\newcommand{\D}{D}
\newcommand{\below}{\preceq}
\newcommand{\safe}{\texttt{safe}}
\newcommand{\map}{\texttt{map}}
\newcommand{\dyn}{\texttt{toDynamic}}
\newcommand{\stat}{\texttt{toStatic}}
\newcommand{\red}{\longrightarrow}
\newcommand{\lub}{\sqcup}
\newcommand{\exFalso}{\texttt{ex-falso}}
\newcommand{\cs}{{\rel{S}}}
\newcommand{\CC}{\texttt{CC}}
\newcommand{\Obj}{\texttt{F}_{\omega o}}
\newcommand{\ObjBase}{\texttt{F}_{\omega}}
\newcommand{\elim}{\texttt{-elim}}
\newcommand{\prot}{\texttt{protect}}
\newcommand{\public}{\texttt{public}}
\newcommand{\unit}{\texttt{<>}}
\newcommand{\trans}{\texttt{trans}}
\newcommand{\join}{\texttt{join}}
\newcommand{\translates}{\hookrightarrow}
\newcommand{\com}{\texttt{com}}
\newcommand{\defined}[2]{#1(#2)^\checkmark}
\newcommand{\notdefined}[2]{#1(#2)^\times}
\newcommand{\refinedBy}{\sqsubseteq}
\newcommand{\TX}{\ensuremath{\tau}}
\newcommand{\pto}{\to_p}
\newcommand{\inl}{\textit{inl}}
\newcommand{\inr}{\textit{inr}}
\newcommand{\MT}{\text{MT}}
\newcommand{\coMT}{\text{coMT}}
\newcommand{\TI}{\text{TI}}
\newcommand{\TS}{\text{TS}}
\newcommand{\Total}{\text{Total}}
\newcommand{\METI}{\text{MEST}}
\newcommand{\udef}{\text{diverge}}
\newcommand{\Alice}{\text{Alice}}
\newcommand{\Bob}{\text{Bob}}
\newcommand{\Charlie}{\text{Charlie}}
\newcommand{\LOW}{\ensuremath{{L}}}
\newcommand{\MED}{\ensuremath{{M}}}
\newcommand{\HIGH}{\ensuremath{{H}}}
\newcommand{\inputify}{\text{inputify}}
\newcommand{\outputify}{\text{outputify}}
\newcommand{\nameify}[1]{\textit{#1}}
\newcommand{\divIfPre}{\nameify{divergeIfPresent}}
\newcommand{\divIfAbs}{\nameify{divergeIfAbsent}}
\newcommand{\id}{\nameify{id}}
\newcommand{\leakAll}{\nameify{leakAll}}
\newcommand{\leakBit}{\nameify{leakBit}}
\newcommand{\terminationLeak}{\nameify{termLeak}}
\newcommand{\divOnLow}{\nameify{divergeOnLow}}
\newcommand{\salaryAverages}{\nameify{salaryAverages}}
\newcommand{\combine}{\nameify{combine}}
\newcommand{\leakLevel}{\nameify{leakLevel}}
\newcommand{\combineAll}{\nameify{combineAll}}
\newcommand{\vlinesep}[2]{\smash{\vrule width .5pt depth #1 height #2}}
\newcommand{\defAs}{\ensuremath{\triangleq}}
\newcommand{\LA}{{\mathbb{L}}}
\newcommand{\divIfEmpty}{\nameify{divergeIfEmpty}}
\newcommand{\nums}{\#}
\newcommand{\n}{\textbf{n}}
\newcommand{\ITE}[3]{\text{if}\ #1\ \text{then}\ #2\ \text{else}\ #3}
\newcommand{\SKIP}{\text{skip}}
\newcommand{\FORK}{\text{fork}}
\newcommand{\JOIN}{\text{join}}
\newcommand{\pc}{\textit{pc}}
\newcommand{\ruleName}[1]{\small\textsc{#1}}
\newcommand{\drSKIP}{\ruleName{[TSKIP]}}
\newcommand{\rSKIP}{\ruleName{TSKIP}}
\newcommand{\drITE}{\ruleName{[TITE]}}
\newcommand{\rITE}{\ruleName{TITE}}
\newcommand{\drFORK}{\ruleName{[TFORK]}}
\newcommand{\rFORK}{\ruleName{TFORK}}
\newcommand{\drASS}{\ruleName{[TASS]}}
\newcommand{\rASS}{\ruleName{TASS}}
\newcommand{\Multef}{\textit{Multef}}
\newcommand{\Labeled}{\textit{Labeled}}
\newcommand{\LIO}{\textit{LIO}}
\newcommand{\labelOf}{\textit{labelOf}}
\newcommand{\LIOlabel}{\textit{label}}
\newcommand{\unlabel}{\textit{unlabel}}
\newcommand{\toLabeled}{\textit{toLabeled}}
\newcommand{\true}{\textit{true}}
\newcommand{\false}{\textit{false}}
\newcommand{\record}{\text{record}}

\title{\bf Dynamic IFC Theorems for Free!}

\ifanon
\author{}
\else
\author{\IEEEauthorblockN{Maximilian Algehed}
\IEEEauthorblockA{
\textit{Chalmers}\\
G\"oteborg, Sweden \\
\texttt{algehed@chalmers.se}}
\and
\IEEEauthorblockN{Jean-Philippe Bernardy}
\IEEEauthorblockA{
\textit{University of Gothenburg}\\
G\"oteborg, Sweden \\
\texttt{jean-philippe.bernardy@gu.se}}
\and
\IEEEauthorblockN{C\u{a}t\u{a}lin Hri\cb{t}cu}
\IEEEauthorblockA{
\textit{MPI-SP}\\
Bochum, Germany \\
\texttt{catalin.hritcu@gmail.com}}
}
\fi

\maketitle

\begin{abstract}
We show that noninterference and transparency, the key soundness theorems for dynamic IFC libraries, can be obtained ``for free'', as direct consequences of the more general parametricity theorem of type abstraction. This allows us to give very short soundness proofs for dynamic IFC libraries such as faceted values and LIO. Our proofs stay short even when fully mechanized for Agda implementations of the libraries in terms of type abstraction.
\end{abstract}

\ifcamera
\fi

\section{Introduction}

The goal of information flow control (IFC) research is to develop
language-based techniques to ensure that security policies relating to
confidentiality and integrity of data are followed, by construction.
This paper is about a recent incarnation of this idea: {\em IFC as a library}.
This appealing approach, pioneered by \citet{li2006encoding} and
championed by Russo~\ETAL~\cite{MAC, LIO, CLIO, HLIO, DCCInHaskell,
FSME, buiras2013library, SecIO, StefanMMR17,vassena2018mac, VassenaBWR16}
among others~\cite{parker2019lweb, Max-DCC, PasquierBS14}, promises to
ease the integration of IFC techniques into existing software development
pipelines, by replacing the specialized languages, compilers, and runtime systems
traditionally needed for IFC applications, with libraries providing
similar guarantees.
Practically speaking, programming with an IFC library is similar to programming
with a specialized IFC language, with one exception:
Rather than being stand-alone, the library integrates with and uses the
features of its host language to provide an interface that guarantees that all
client programs are secure.

These libraries enforce IFC using two key language features:
\begin{itemize}
  \item[] {\bf Controlling side effects:} most IFC libraries are implemented in
    safe \cite{SafeHaskell} Haskell~\cite{Haskell}, a language
    that allows the library author to enforce type-based control over side effects
    (with the exception of non-termination).
  \item[] {\bf Abstraction:} all IFC libraries rely on type abstraction to provide data
    confidentiality and integrity.
\end{itemize}
But embedding IFC as a library risks leaving a gap in the soundness proofs,
which usually do not cover how the library interacts with the host language.
Indeed, the typical soundness proofs work by constructing a \emph{model} of the
library as a more-or-less standard IFC calculus, for which the authors prove
some variant of the noninterference security property~\cite{Denning}.
So although these noninterference proofs are sometimes mechanically
verified~\cite{StefanMMR17,vassena2018mac,FSME,parker2019lweb,Max-DCC}, they give no
guarantees about the realization of the calculus \emph{as a library}.

The formal connection between the calculus and the library, which relies on host
language type abstraction, was so far almost never investigated.
%
%
We are the first to provide formal proofs {\em explicitly covering the use
  of type abstraction for implementing dynamic IFC libraries}, such as LIO.
%
%
Moreover, the key to scaling security proofs to other advanced features of
dynamic IFC libraries is employing better proof techniques.


In this paper we provide simpler proofs for Agda~\cite{Agda} implementations in
terms of type abstraction of two different kinds of dynamic IFC libraries.
On the one hand, we study the sequential LIO library~\cite{LIO}, in
which individual values can be labeled with metadata specifying their
confidentiality and integrity levels and computations carry a ``current label''
that soundly over-approximates the level of already inspected labeled values.
On the other hand, we study a library based on {\em faceted
values}~\cite{austin2014typed,FSME}, which are decision trees that can evaluate
to different values based on the privilege level of the observer.
These two styles reflect the most common ways to enforce dynamic IFC as a library.

To prove noninterference in a simple way, we give semantics to the libraries in
terms of {\em logical relations}.
Every type $T$ induces a binary \iffull logical \fi relation $\rel T$, such that every
well-typed program $p : T$ is related to itself with respect to this relation.
This connection between terms, types, and logical relations is
called the \emph{fundamental lemma} of logical relations, or the
\emph{abstraction theorem}~\cite{reynolds_types_1983}, or
\emph{parametricity}.
In his seminal work, \citet{reynolds_types_1983} uses this technique to show
that users of an abstract type can never observe the details of its
implementation.
In this paper we apply this idea to dynamic IFC, showing that noninterference
for dynamic IFC libraries is a direct consequence of the same parametricity
theorem.\footnote{Theorems directly obtained from parametricity are often called
  {\em free theorems}~\cite{WadlerParametricity}, hence the title of this paper.}

In practice, we implement the dynamic IFC libraries in a language with dependent
types.
This allows us to program our libraries and their clients in the same
formalism we use to reason about the security of such programs.
We also use the fact that (dependent) types are turned mechanically to
logical relations and programs are turned mechanically to proofs of
satisfaction of such relations, and everything remains neatly
expressible within the framework of a language with
dependent types~\cite{ParametricityJFP}.
This way, we do not need to prove any fundamental lemma for our
logical relations, and additionally the junction between
implementation and theory is watertight.

This proof technique was recently applied by
\citet{SimpleNoninterferenceFromParametricity} to show the noninterference of
\emph{static} IFC libraries.
The work in this paper differs in two crucial ways:
\begin{enumerate}
  \item We show that the same technique can be applied to prove noninterference
    for {\em dynamic} IFC libraries. While this might seem counter-intuitive at first, since
    parametricity is a property of type abstraction, which is a static
    enforcement mechanism, our work shows formally that dynamic IFC libraries
    ultimately achieve their security also from type abstraction.
%
  \item We show that the same proof technique can be used to prove more than
    just noninterference.
    Specifically, we use the parametricity theorem to also prove that our faceted
    values library is {\em transparent}~\cite{MF, FSME}, ensuring that programs
    that are already secure do not have their semantics altered by our library.
%
\end{enumerate}

More importantly, we inherit the simplicity of the parametricity-based
proof technique.
In particular, our completely mechanized proof of noninterference for
a type abstraction-based implementation in Agda~\cite{Agda}
of the state-of-the-art LIO library is an order of magnitude shorter
than the previous partial Coq proof for
a model of the same library~\cite{StefanMMR17}.
Moreover, the simplicity of our proof technique allows us to cover \emph{more} of
our implementation of the LIO library than the previous Coq proof.
Specifically, our mechanized LIO proof also covers the non-trivial mechanisms that enable
soundly recovering from IFC and user-thrown exceptions~\cite{StefanMMR17, hritcu2013all}.

Concretely, we make the following {\bf contributions}:
\begin{itemize}
  \item We give Agda implementations of two core dynamic IFC libraries:
    LIO~\cite{LIO, StefanMMR17} (\cref{ni-lio}) and a faceted values library
    inspired by the Faceted \cite{austin2014typed} and Multef~\cite{FSME}
    Haskell libraries (\cref{DIFC-as-a-library}).
  \item We use the parametricity theorem (\cref{sec:parametricity}) for these
    dynamic IFC libraries to provide simple noninterference proofs for any client of
    each library (\cref{sec:NIproofs}).
    %
    %
    %
  \item We use the same proof technique to also show that our faceted values
    library satisfies transparency~\cite{MF, FSME}, which ensures that
    this library does not alter the semantics of already secure programs.
    (\cref{sec:transparency}).
    %
  \item Finally, all our proofs have been fully mechanized in the Agda proof
    assistant, and are available as supplementary material for this \ifanon
    submission.\else paper.\footnote{Supplementary material available at
      \href{https://github.com/MaximilianAlgehed/DynamicIFCTheoremsForFree}
           {\tt https://github.com/\\ MaximilianAlgehed/DynamicIFCTheoremsForFree}}\fi{}
    As a consequence of the use of parametricity for dependent types,
    the implementation of our libraries and the proofs are surprisingly short:
    our whole formal development for the two different libraries is less than
    one thousand lines of Agda.
\end{itemize}

While we use Agda as our language of choice in this paper, since it has a well
worked out theory of parametricity (a generalization of
\autoref{thm:abstraction} was already proved by \citet{ParametricityJFP}), we
expect that
our results are easily to apply to other languages with strong
abstraction mechanisms and dependent types, like Coq and F$^\star$.


\section{\label{DIFC-as-a-library}Dynamic IFC as a Library}
\begin{figure}
  \centering
\begin{agda}
module MultefImplementation where
  data Fac : Set → Set where
    return : {A : Set} → A → Fac A
    facet  : {A : Set} → Label → Fac A → Fac A
                       → Fac A

  bind : {A B : Set} → Fac A → (A → Fac B) → Fac B
  bind (return a) c = c a
  bind (facet ℓ f₀ f₁) c = facet ℓ (bind f₀ c) (bind f₁ c)
\end{agda}
  \caption{\label{fig:multef} Faceted values part of Multef in Agda}
\end{figure}

Before turning our attention to embedding dynamic IFC as a library, we give
a quick primer on our host language, Agda: a total functional programming
language with dependent types~\cite{Agda}.

One characteristic of Agda is that terms, types, and propositions are unified.
Thus, a single dependent arrow type-former (\ai{->}) can be used to define
function types and universal quantification.
For example the type \ai{(x : Bool) -> Bool} is a function type whose domain and
co-domain are Booleans.
The type \ai{(x : Bool) -> x ≡ true ∨ x ≡ false} is an example of a
quantified proposition over booleans.
The type \ai{(A : Set) → List A → List A} is that of a polymorphic
list-transformation.
This last example illustrates quantification over types, which is
seen as a dependent function whose domain is the type of types,
written \ai{Set} in Agda.
(To avoid logical inconsistencies sets of sets are organized in a tower, such
that \ai{Setᵢ : Setᵢ₊₁}, and \ai{Set} is a shorthand for \ai{Set₀}.)

As a convenience feature, Agda offers \emph{implicit arguments}.
If a function domain is written with braces, for example \ai{f : {A : Set} →
List A → List A}, then when calling the function \ai{f} one will omit the
corresponding argument, for example \ai{f someList : List A}.
In this situation, Agda will try to infer the omitted argument from the
context. For example if \ai{someList : List Bool} then \ai{A = Bool}.
If doing so is impossible or ambiguous, Agda will report an error.

While Agda has many more features, we will only use a small
time-tested subset.
We also make use of records, which will be discussed below, and some
limited use of data types, that work like in other functional languages.

Our first library is based on \emph{faceted values}
\cite{austin2014typed,MF} and is displayed in \Cref{fig:multef}.
This is a straightforward port of the faceted values part of
the Multef Haskell library~\cite{FSME} to Agda.
Faceted values are binary decision trees that can evaluate to different values
based on the privilege level of the observer.
Formally, a faceted value \ai{f : Fac A} can either be a regular value \ai{v :
A} that does not depend on who is observing it, in which case \ai{f = return
v}, or it can depend on who is observing it, in which case it is a tree node
containing a label \ai{ℓ : Label} (we assume a base-type \ai{Label} of security
labels explained below) and children \ai{f₀} and \ai{f₁} of type \ai{Fac A}
(\IE\ \ai{f = facet ℓ f₀ f₁}).
In the code in \Cref{fig:multef}, this tree structure is formalised by the
\ai{data … where …} declaration.
Similarly to other functional languages like Haskell, it says that
\ai{Fac : Set → Set} is a type-constructor mapping types to types, and that it has two
constructors \ai{return} and \ai{facet}.

Returning to the meaning of \ai{facet ℓ f₀ f₁}, if an observer has access to
level \ai{ℓ} they will see \ai{f₀}, otherwise \ai{f₁}.
For example, the faceted value \ai{facet Alice (return 0) (return 1) : Fac Int} looks like
\ai{0} to anyone who is able to observe values with \ai{Alice}'s
confidentiality level, and looks like \ai{1} to everyone else.

Combining two faceted values also yields a faceted value.
For example, if we add \ai{f₀ = facet Alice (return 0) (return 1)} and
\ai{f₁ = facet Bob (return 2) (return 3)}, we get the \emph{recursive} faceted value:
\begin{agda}
f₀+f₁ = facet Alice
         (facet Bob (return (0 + 2)) (return (0 + 3)))
         (facet Bob (return (1 + 2)) (return (1 + 3)))
\end{agda}
By inspecting this value, we see that if the observer can observe \emph{both}
\ai{Alice}'s and \ai{Bob}'s data, they observe the sum \ai{f₀ + f₁} as
\ai{0 + 2}, whereas an observer who can see only \ai{Bob}'s data, \emph{not}
\ai{Alice}'s, sees the sum as \ai{1 + 2}.

To formally define this addition operation on faceted values we follow the
literature \cite{FSME,austin2014typed} and show that faceted values form a
\emph{monad} \cite{moggi1991notions,wadler1995monads}, which provides a general
computational tool we can use to easily define operations like addition.
A monad is a triple of a type former, \ai{M : Set → Set}  (here \ai{M} = \ai{Fac}), and two
operations, \ai{return : {A : Set} → A → M A} that takes a pure value and
embeds it into a monadic value, and \ai{bind : {A B : Set} → M A → (A → M B) →
M B} that takes a monadic value \ai{m : M a} and a continuation \ai{c : A →
M B} and produces a new monadic \ai{bind m c : M B}.\footnote{A
monad requires \ai{bind} and \ai{return} to satisfy certain \emph{algebraic laws}.
However, these are orthogonal to our development and thus omitted.}

In \Cref{fig:multef}, \ai{return} is a constructor of the \ai{Fac} type, and \ai{bind}
is defined as a function using pattern-matching.
If \ai{f = return a} then \ai{bind f c = c a}, and if \ai{f = facet ℓ f₀ f₁} then
\ai{bind f c} is defined recursively as \ai{facet ℓ (bind f₀ c) (bind f₁ c)}.
In effect, this means that in \ai{bind f c} the continuation \ai{c} is computed
once for every leaf of the faceted tree \ai{f}, potentially producing bigger
faceted trees to replace the old leaves.

Using the monad operations, we can define addition and similar operations for
\ai{Fac} easily in the following manner:
\begin{agda}
(+) : Fac Int → Fac Int → Fac Int
fx + fy = bind fx λ x →
          bind fy λ y →
          return (x + y)
\end{agda}
The syntactic form \ai{λ x → …} in the code above is simply Agda
notation for lambda abstraction.
%

As illustrated by the addition example above,
the \ai{bind} operation from \cref{fig:multef} works by traversing faceted
values in its operands and constructing a recursive faceted value, taking into
account all possible observers.

Following the IFC literature \cite{sabelfeld2003language},
whether a level can be observed from another level is given by a
partial order \ai{_⊑_ : Label → Label → Bool}.
This order is also required to form a join semi-lattice, and thus it has a
least element \ai{ℓ⊥ : Label} and the least-upper-bound always exists and is
given by the function \ai{_⊔_ : Label -> Label -> Label}.
Using this lattice structure, we can define \ai{project l f},
intuitively what a value of type \ai{f : Fac a} will ``look like'' to
an observer at level \ai{l}.
We define \ai{project l f} by recursion on~\ai{f}:
\begin{agda}
project : {A : Set} → Label → Fac A → A
project ℓ (return a)       = a
project ℓ (facet ℓ' f₀ f₁) =
  &\kw{if}& ℓ' ⊑ ℓ &\kw{then}&
    project ℓ f₀
  &\kw{else}&
    project ℓ f₁
\end{agda}

In turn, from the definition of \ai{project}, we can define what it means for a
program, for example a function \ai{p : Fac Int → Fac Int}, to be secure.
The program \ai{p} is \emph{noninterfering} if given any label \ai{ℓ : Label}
and two faceted values \ai{f₀, f₁ : Fac Int} such that \hsi{project ℓ f₀ ≡ project ℓ f₁} we
have that \ai{project ℓ (p f₀) ≡ project ℓ (p f₁)}.
In other words, \ai{p} does not reveal information from a different security
level to an observer at level \ai{ℓ}.

\begin{figure}
\begin{agda}
module Temp where

data HotOrCold : Set where
  Hot  : HotOrCold
  Cold : HotOrCold

isCold : Fac Int → Fac HotOrCold
isCold fint = bind fint λ x →
  &\kw{if}& x > 25 &\kw{then}&
    return Hot
  &\kw{else}&
    return Cold
\end{agda}
  \caption{\label{fig:Temp} A client of the Multef library}
\end{figure}

To see an example of this property in action, consider the \hsi{Temp} client module in \cref{fig:Temp}.
When we give \hsi{isCold} the arguments \hsi{f₀} and \hsi{f₁} defined as:
\begin{agda}
f₀ f₁ : Fac Int
f₀ = facet Alice (return 10) (return 0)
f₁ = facet Alice (return 30) (return 0)
\end{agda}
we get:
\begin{agda}
isCold f₀ = facet Alice (return Cold) (return Cold)
isCold f₁ = facet Alice (return Hot) (return Cold)
\end{agda}
If the observer level \ai{ℓ} is \ai{Bob}, who cannot see \ai{Alice}'s data
(\IE{} \ai{Bob ̸⊑ Alice}) we have that
\ai{project f₀ Bob ≡ project f₁ Bob} and also
\ai{project (isCold f₀) Bob ≡ project (isCold f₁) Bob}.

\begin{figure}
\begin{agda}
isAlicePos : Fac Int → Fac Bool
isAlicePos (facet Alice (return n) f) = return (n > 0)
isAlicePos f = return False
\end{agda}
  \caption{\label{fig:Bad} An ill-behaved client}
\end{figure}

The goal of library-based IFC is to ensure that all client code behaves
securely, as \ai{isCold} does.
However, suppose that we could write the code in \cref{fig:Bad}.
Function \ai{isAlicePos} uses pattern matching to check whether its faceted
input is precisely of the form \ai{facet Alice (return n) f}, for some \ai{n}
and \ai{f}, and if so returns a raw (un)faceted value \ai{return (n > 0)}.
Otherwise \ai{isAlicePos} returns \ai{return False}.
The function \ai{isAlicePos} clearly breaks noninterference.
What goes wrong here is that even though we carefully ensure that the functions
\ai{facet}, \ai{return}, and \ai{bind} respect noninterference, the
\ai{isAlicePos} client function could just side-step our interface and break up
faceted values to look directly at \ai{Alice}'s secrets.

To ensure security, therefore, we must use the abstraction mechanisms of the
underlying programming language \cite{cardellimodules}.
In the case of Agda, this means two things.
First, we must define the 
\emph{interface} of the IFC library,
as a record type.
In the case of core Multef, the interface can be found in \cref{fig:multef:interface}.

\begin{figure}
\begin{agda}
record MultefInterface where
  field
    Fac : Set → Set
    return : {A : Set} → A → Fac A
    facet  : {A : Set} → Label → Fac A → Fac A
                       → Fac A
    bind : {A B : Set} → Fac A → (A → Fac B)
                       → Fac B
\end{agda}
  \caption{\label{fig:multef:interface} The abstract interface to the Multef library}
\end{figure}

Second, any client can depend \emph{only} on this interface.
For Agda, this is ensured by parameterizing the client modules by the IFC library interface,
for example, the following will make the \ai{BadClient} module parametric:
\begin{agda}
module BadClient (imp : MultefInterface) where
open MultefInterface imp

isAlicePos : Fac Int → Fac Bool
…
\end{agda}
Then Agda will not allow us compile this module because of the \ai{isAlicePos}
function: we get an error telling us that \ai{Fac} is an abstract type, and
does not expose a constructor \ai{facet} on which we can pattern-match (in the
interface, \ai{facet} is just a \emph{function} and not a datatype constructor).
%
%
Finally, we must make sure that our secure implementation,
which we will call \ai{sec}, indeed implements the interface.
In Agda syntax:
\begin{agda}
sec : MultefInterface
sec = record {MultefImplementation}
\end{agda}
In \cref{sec:NIproofs}, we will use this abstraction barrier to provide a
\emph{relational interpretation} for the Multef interface (in addition to that
of the LIO library \cite{LIO}) and a proof that the implementation satisfies
this relational interpretation.
By the abstraction theorem, all client modules will be proved to be
noninterfering.


\section{\label{sec:parametricity}Parametricity and Data Abstraction}
Before we dive into proving noninterference for our libraries,
we give a quick primer on constructive reasoning in Agda.
Using the propositions-as-types (also known as the Curry-Howard) correspondence
\cite{howard1980formulae, Wadler15}, we can interpret Agda types as propositions and
programs as proofs in constructive logic:

The type of types is \ai{Set} and inhabitants of this type can also be seen as
propositions.
Type \ai{⊤ : Set} is inhabited by a trivial \ai{tt : ⊤} inhabitant, and thus
represents Truth as a proposition:
\begin{agda}
data ⊤ : Set where tt : ⊤
\end{agda}
Conversely \ai{⊥ : Set} is not inhabited by any term and thus represents Falsity.
Conjunction \ai{_∧_ : Set → Set → Set} and disjunction \ai{_∨_ : Set → Set → Set}
are implemented as product and sum types, respectively.
Finally, as also explained in the previous section, quantification and
implication correspond to the dependent function type \ai{(x : A) → B}.

Since types are propositions, their inhabitants are proofs.
Concretely, if \ai{P : Set} is a proposition (type) then
\ai{t : P} is a proof (inhabitant) of the proposition (type).
For example, the canonical proof of proposition \ai{A ∧ B → B ∧ A} is the
following function:
\begin{agda}
swap : A ∧ B → B ∧ A
swap (a , b) = (b , a)
\end{agda}

With this background in place, we now introduce parametricity and data
abstraction for dependently typed languages.
This proof technique is based on logical relations~\cite{plotkin1973logical,
  Statman85}, which are an elegant tool to prove properties about programming
languages, and in particular IFC.
The key idea is that one interprets every type as a relation.
For every type \ai{A}, one builds a relation \ai{⟦A⟧} ---
thus \ai{⟦A⟧ a₀ a₁} is a proposition given two values \ai{a₀ a₁ : A} and
so \ai{⟦A⟧ : A → A → Set}.
%
One then proves the fundamental lemma of logical relations, also known as
abstraction theorem, or parametricity theorem:

\begin{proposition}[Parametricity]
  \label{prop:parametricity}
  If \ai{t : A} then \ai{⟦A⟧ t t}.
  That is, every program \ai{t} of type \ai{A} satisfies the relational
  interpretation of its type \ai{⟦A⟧}.
\end{proposition}
One often uses a custom logical relation \cite{NoninterferenceForFree},
but there is a general, most fundamental way to interpret dependent types as
relations, given by \citet{ParametricityJFP}, which we adapt below for the
syntax of Agda types.\footnote{Their theory is for pure type systems with
  inductive families, covering all the features of Agda that we use in this paper.}

\begin{definition}(Relational interpretation of types)
  \label{def:rel-types}
  This meta-level definition works by induction on the structure of types. 
  \begin{align*}
    \rel {\ai {Setᵢ}}\ A₀\ A₁ & = A₀ → A₁ → \ai{Setᵢ}\\
    \rel {\aim {(x:\(A\)) → \(B\)}}\ f₀\ f₁ & =\aim{(x₀ : \(A₀\))}\\
                                            &\aim{→ (x₁ : \(A₁\))}\\
                                            &\aim{→ (xᵣ : \(\rel A\) x₀ x₁)}\\
                                            &\aim{→ \(\rel B\)\ \((f₀\ x₀)\ (f₁\ x₁)\)}\\
    \rel {\aim{record field fⁱ} : Aᵢ}\ r₀\ r₁ & = \aim {record field}\\
                                              &\ \ \ \ \ \ \aim{f}^i_r : \rel{Aᵢ}\ (\aim{fⁱ}\ r₀)\ (\aim{fⁱ}\ r₁)\\
    \rel {\textbf{B}}\ b₀\ b₁         & = \aim{b₀ ≡ b₁}
  \end{align*}
\end{definition}
As mentioned above, the type of types (\ai{Set}) is interpreted as a function from two types
to \ai{Set} (\IE a relation).
A function type is interpreted as a relation requiring that inhabitants map related
arguments to related results.
In particular, the dependent function type \ai{(x : A) → B} binds \ai{x} as a
variable in \ai{B}.
The single bound variable \ai{x} is turned into three bound variables in the translated type,
\ai{x₀} and \ai{x₁} that bind elements of \ai{Aᵢ} respectively, and \ai{xᵣ} that binds a proof
that \ai{x₀} and \ai{x₁} are related by \ai{⟦A⟧}.
A record type is interpreted as a relation relating two instances of the record
by relating all their fields.
Finally, base-types (\textbf{\ai{B}}), like booleans, are interpreted as
propositional equality.
In short, propositional equality at the type \ai{A} is a type \ai{≡} with a
single inhabitant \ai{refl : {a : A} → a ≡ a}.
Technically, this means that Agda will only equate two terms \ai{a₀, a₁ : A} if
it can prove that they both reduce to some term \ai{a : A}.

To prove the fundamental lemma, one proceeds by giving a
relational interpretation $\rel t$ for every term $t$:
In the development of \citet{ParametricityJFP} that we use in this paper, this
interpretation is as follows.
\begin{definition}{Relational interpretation of terms}
  \label{def:rel-terms}
\begin{align*}
  \rel {t\ u}                    & = \rel t\, u₀\, u₁\, \rel u \\
  \rel {\aim{λx → \(t\)}}        & = \aim {λx₀ → λx₁ → λxᵣ →} \rel t \\
  \rel{\aim{record \{fⁱ = tᵢ\}}} & = \aim{record \{f}^i_r\aim{ = \(\rel {tᵢ}\)\}} \\
  \rel{\aim{fⁱ}\ t}              & = \aim{f}^i_r\ \rel{t}\\
  \rel {\aim{x}}                 & = \aim{xᵣ}\\
  \rel {\textbf{b}}              & = \aim{refl}\\
  \rel A                         & = \aim {λa₀ → λa₁ →} \rel A\ \aim{a₀}\ \aim{a₁}
\end{align*}
\end{definition}

This interpretation mimics the behavior of the relational interpretation
of types.
In particular, if the term is a base type constant (the penultimate case) then
the proof of relatedness is simply reflexivity of equality, and if the term is
a type (the last case), we construct an explicit relation and fall back to the
interpretation for types.
Thus, the interpretation of types as relation and the
interpretation of terms as proofs can be unified, hence the use of a
single notation $⟦·⟧$ for both purposes.
Moreover, the translation of function types is mimicked in the translation of lambda
terms, a single bound \ai{x} is turned into three bound \ai{x₀}, \ai{x₁}, and \ai{xᵣ}.
Finally, each use of a variable \ai{x} is turned into a use of the bound proof
\ai{xᵣ} that \ai{x₀} and \ai{x₁} are appropriately related.
We refer the reader to \citet{ParametricityJFP} for details.

This relational interpretation of terms and types provides an once and for all
proof of the parametricity theorem:
\begin{theorem}(Parametricity \cite{ParametricityJFP})
  \label{thm:abstraction}~\\
  If \ai{t : A} then \ai{⟦t⟧ : ⟦A⟧ t t}.
\end{theorem}
%
%

As an illustration, we show how to use \cref{thm:abstraction} to prove properties about
an abstract module and its clients.
Consider the following (restricted) interface for Booleans:
\begin{samepage}
\begin{agda}
  record Booleans where
    field
      Bool  : Set
      true  : Bool
      false : Bool
      ∧     : Bool → Bool → Bool
\end{agda}
\end{samepage}
The above declares an abstract type, \ai{Bool}, two
constants \ai{true} and \ai{false} of type \ai{Bool},
and a binary operation $\wedge$ over \ai{Bool}.
We instantiate this interface in the standard way:
\begin{agda}
module Impl where
  data Bool : Set where
    true  : Bool
    false : Bool

  ∧ : Bool → Bool → Bool
  ∧ false _ = false
  ∧ _ b     = b

booleans : Booleans
booleans = record {Impl}
\end{agda}
What may be surprising about this \ai{Booleans} interface is that
when we use it, we always end up writing monotonic functions.
More precisely, if we write a function
\ai{o : (imp : Booleans) → Bool imp → Bool imp}, it is
possible to prove, using parametricity, that if \ai{b₀, b₁ : Impl.Bool} are
such that \ai{b₀} implies \ai{b₁} then \ai{o booleans b₀}
implies \ai{o booleans b₁}.
Intuitively, this is because the \ai{Booleans} interface gives the
\ai{o} function no way to do negation.
This means that all \ai{o imp} can do to its \ai{b : Bool imp} argument is
to either discard it and return some other boolean or
take its conjunction with some other boolean.
These other booleans are either constants, or obtained by calling functions,
which are themselves parametric in \ai{imp}.
However, ``by induction'', these functions are also monotonic and so \ai{o} is
monotonic.

Making this ``by induction'' phrase precise demands an argument based on logical relations.
In Agda we can let \cref{thm:abstraction} do the ground work, making
the proof feel nearly automatic.
To see how this works formally, we need to understand two things.
Firstly, the \emph{standard} relational interpretation of the Booleans
interface ($\rel{\aim{Booleans}}$, obtained mechanically by using the
meta-level function from \cref{def:rel-types}), tells us how to relate two
implementations of $\aim{Booleans}$:\footnote{Accessing a record field is done
by treating the field name as a function with the analyzed record as an
additional first argument.}
\begin{agda}
record ⟦Booleans⟧ (m₀ m₁ : Booleans) : Set₁ where
  field
    Boolᵣ  : Bool m₀ → Bool m₁ → Set
    trueᵣ  : Boolᵣ (true m₀) (true m₁)
    falseᵣ : Boolᵣ (false m₀) (false m₁)
    ∧ᵣ     : ∀ a₀ a₁ → Boolᵣ a₀ a₁ →
             ∀ b₀ b₁ → Boolᵣ b₀ b₁ →
             Boolᵣ (∧ m₀ a₀ b₀) (∧ m₁ a₁ b₁)
\end{agda}
This relation contains a \emph{custom} logical relation (\ai{Boolᵣ}), such that
each method in the interface respects this relation (and \ai{trueᵣ, falseᵣ, ∧ᵣ}
are proofs witnessing this).

Secondly, because \ai{o} is parameterized by \ai{imp : Booleans}, we have that
\ai{⟦o⟧} is parameterized over \ai{⟦Booleans⟧}:
\begin{agda}
⟦o⟧ : (imp₀ imp₁ : Booleans)
   → (impᵣ : ⟦Booleans⟧ imp₀ imp₁)
   → (b₀ : Bool imp₀) (b₁ : Bool imp₁)
   → (bᵣ : Boolᵣ impᵣ b₀ b₁)
   → Boolᵣ impᵣ (o imp₀ b₀) (o imp₁ b₁)
\end{agda}
Now all it takes to prove our theorem is to realize that:
\begin{enumerate}
  \item We care about two \ai{Bool booleans}, so we have to take \ai{imp₀ = imp₁ = booleans}, and
  \item \ai{impᵣ : ⟦Booleans⟧ booleans booleans} is an argument to \ai{⟦o⟧} that
    we get to pick, and
  \item the final thing we want to prove is that if
    \ai{b₀ ⇒ b₁}, then \ai{o … b₀ ⇒ o … b₁},
    so we want that \ai{Boolᵣ b₀ b₁ = b₀ ⇒ b₁}.
\end{enumerate}
With insight (2) and (3) and a definition of $⇒$ as a relation on \ai{Bool}:
\begin{agda}
  _⇒_ : Bool -> Bool -> Set
  true ⇒ false = ⊥
  _    ⇒ _     = ⊤
\end{agda}
we can define the fields of \ai{booleansᵣ} as follows:
\begin{agda}
booleansᵣ : ⟦Booleans⟧ booleans booleans
booleansᵣ = record {
    Boolᵣ  = _⇒_
  ; trueᵣ  = tt
  ; falseᵣ = tt
  ; ∧ᵣ true  true aᵣ b₀ b₁ bᵣ = bᵣ
    ∧ᵣ false a₁   aᵣ b₀ b₁ bᵣ = tt
  }
\end{agda}
The proofs that \ai{true} and \ai{false} satisfy the relation are trivial.
For \ai{∧}, we proceed by a simple case analysis.
This gives us all the pieces we need to construct our monotonicity proof:
\begin{agda}
⟦o⟧ booleans booleans booleansᵣ :
    (b₀ : Bool booleans) → (b₁ : Bool booleans)
   → (bᵣ : b₀ ⇒ b₁) → o booleans b₀ ⇒ o booleans b₁
\end{agda}

Now, suppose we add negation to the interface of \ai{Booleans}, with the usual
implementation in \ai{booleans}:
\begin{agda}
record Booleans : Set₁ where
  field
    …
    neg : Bool -> Bool

module Impl where
  …
  neg : Bool → Bool
  neg true = false
  neg false = true

booleans : Booleans
booleans = record {Impl}
\end{agda}
If we try to prove the same monotonicity theorem, which now shouldn't hold, we run into
issues when we try to provide \ai{negᵣ} in the new \ai{booleansᵣ}.
Specifically, trying to fulfill the proof obligation by case analysis leaves us with an impossible goal:
\begin{agda}
negᵣ : (a₀ a₁ : Bool) → a₀ ⇒ a₁
    → neg booleans a₀ ⇒ neg booleans a₁
negᵣ true  true  tt = tt
negᵣ false true  tt = ? -- Goal is true ⇒ false
negᵣ false false tt = tt
\end{agda}
%


As demonstrated by this example, parametricity for dependent types is not just
an adequate tool for reasoning about meta-theoretic properties of libraries,
it's also compositional.
To prove that a library guarantees some property, one
simply defines the necessary relations on one's types and proves that each
operation preserves these relations.
In \cref{sec:NIproofs} we use this technique to show
noninterference for two security libraries.
The proofs work like the simple proof above: we define the relations necessary
to prove noninterference and show that each operation in the library respects
the relations.


\section{\label{sec:NIproofs}Two Proofs of Noninterference}
In this section we use the parametricity technique outlined above to prove
noninterference for two dynamic IFC libraries.
The first proof (\cref{ni-multef}) is for the faceted values part of the Multef
Haskell library, which we have already introduced in \cref{DIFC-as-a-library}.
The second proof is for the significantly more complex LIO library (\cref{ni-lio}).

\subsection{Noninterference for Faceted Values}
\label{ni-multef}

\begin{figure*}
  \begin{multicols}{2}
\begin{agda}
Facᵣ  : (A₀ A₁ : Set)
     → (Aᵣ : ⟦Set⟧ A₀ A₁)
     → ⟦Set⟧ (Fac sec A₀) (Fac sec A₁)

facetᵣ : (A₀ A₁ : Set) → (Aᵣ : ⟦Set⟧ A₀ A₁)
      → (ℓ₀ ℓ₁ : Label) → (⟦ℓ⟧ : ⟦Label⟧ ℓ₀ ℓ₁)
      → (f₀₀ : Fac sec A₀) → (f₀₁ : Fac sec A₁)
      → (f₀ᵣ : Facᵣ A₀ A₁ Aᵣ f₀₀ f₀₁)
      → (f₁₀ : Fac sec A₀) → (f₁₁ : Fac sec A₁)
      → (f₁ᵣ : Facᵣ A₀ A₁ Aᵣ f₁₀ f₁₁)
      → Facᵣ A₀ A₁ Aᵣ (facet sec ℓ₀ f₀₀ f₁₀)
                       (facet sec ℓ₁ f₀₁ f₁₁)
\end{agda}
  \break  
\begin{agda}
returnᵣ : (A₀ A₁ : Set) → (Aᵣ : ⟦Set⟧ A₀ A₁)
       → (a₀ : A₀) → (a₁ : A₁) → (aᵣ : Aᵣ a₀ a₁)
       → Facᵣ A₀ A₁ Aᵣ (return sec a₀) (return sec a₁)

bindᵣ : (A₀ A₁ : Set) → (Aᵣ : ⟦Set⟧ A₀ A₁)
     → (B₀ B₁ : Set) → (Bᵣ : ⟦Set⟧ B₀ B₁)
     → (f₀ : Fac sec A₀) → (f₁ : Fac sec A₁)
     → (fᵣ : Facᵣ A₀ A₁ Aᵣ f₀ f₁)
     → (c₀ : A₀ → Fac sec B₀) → (c₁ : A₁ → Fac sec B₁)
     → (cᵣ : (a₀ : A₀) → (a₁ : A₁) → (aᵣ : Aᵣ a₀ a₁)
           → Facᵣ B₀ B₁ Bᵣ (c₀ a₀) (c₁ a₁))
     → Facᵣ B₀ B₁ Bᵣ (bind sec f₀ c₀) (bind sec f₁ c₁)
\end{agda}
  \end{multicols}
  \caption{\label{fig:RMultef} Fields in \ai{⟦MultefInterface⟧ sec sec}}
\end{figure*}

Recall the faceted values interface in \cref{fig:multef:interface} from \cref{DIFC-as-a-library}.
It exports a type \ai{Fac} for faceted values and operations \ai{facet},
\ai{return}, and \ai{bind} for manipulating them.
The goal in this section is to show that any client library of this abstract
interface obeys noninterference.

Intuitively, this means that any client function of \ai{MultefImplementation}
needs to take \ai{ℓ}-equivalent inputs to \ai{ℓ}-equivalent outputs.
In order to make the above statement explicit, we recall the \ai{sec} instantiation of
\ai{MultefInterface} and the definition of \ai{project} from \cref{DIFC-as-a-library}
and provide the following definition of \ai{ℓ}-equivalence:
If \ai{A} is a base-type and \ai{f₀, f₁ : Fac A} we say that
\ai{f₀} and \ai{f₁} are \ai{ℓ}-equivalent, written \ai{f₀ ∼⟨ ℓ ⟩ f₁}, when:
\begin{agda}
f₀ ∼⟨ ℓ ⟩ f₁ = project f₀ ℓ ≡ project f₁ ℓ
\end{agda}
Where \ai{≡} is propositional equality.

In order to prove noninterference, we need to prove that for a given
base-type $\aim{A}$ (say \ai{Bool}), function \ai{o : Fac A → Fac A} and
faceted values \ai{f₀, f₁ : Fac A} such that \ai{f₀ ∼⟨ ℓ ⟩ f₁} we have that
\ai{o f₀ ∼⟨ ℓ ⟩ o f₁}.
However, this proposition only holds if \ai{o} is a function in a \emph{client}
of the \aim{Multef} library (more accurately, its abstract interface).
In other words, noninterference only needs to hold for a function
\ai{o : (m : MultefInterface) → Fac m A → Fac m A}.

Recall from the \ai{Booleans} example in \cref{sec:parametricity} that we can
reason about \ai{o} by providing a suitable \ai{secᵣ : ⟦Multef⟧ sec sec}, a
\emph{proof} that \ai{sec} (from \cref{DIFC-as-a-library}):
\begin{agda}
sec : MultefInterface
sec = record {MultefImplementation}
\end{agda}
satisfies the relational interpretation of its type.
Formally, parametricity requires us to construct a \ai{fᵣ : ⟦A⟧ f f} for each function
\ai{f : A} in the \ai{MultefInterface} record type.
Concretely, this means that to construct \ai{secᵣ}, we need to construct Agda
terms inhabiting the four types in \cref{fig:RMultef}.

Picking the implementation of \ai{secᵣ : ⟦MultefInterface⟧ sec sec} that we
will use in our noninterference proof is straightforward given our formulation of
the \ai{f₀ ∼⟨ ℓ ⟩ f₁} relation above.
Given \ai{ℓ*} as the attacker-level that we are concerned about, we pick:
\begin{agda}
Facᵣ A₀ A₁ Aᵣ f₀ f₁ = Aᵣ (project f₀ ℓ*) (project f₁ ℓ*)
\end{agda}
Note that if \ai{A} is a \emph{base type}, like \ai{Bool}, then \ai{⟦A⟧ = _≡_}
and so \ai{Facᵣ A A ⟦A⟧} corresponds to \ai{∼⟨ ℓ* ⟩}.
In other words, \ai{Facᵣ Bool Bool ⟦Bool⟧ f₀ f₁} is equivalent to \ai{f₀ ∼⟨ ℓ* ⟩ f₁}.
The definitions of \ai{facetᵣ}, \ai{returnᵣ}, and \ai{bindᵣ} are easy to fill
out, and can be looked up in the Agda mechanization.\footnote{It may be helpful
to recall that \ai{Label} is a base-type, so the \ai{⟦Label⟧ ℓ₀ ℓ₁} argument to \ai{facetᵣ}
is equivalent to \ai{ℓ₀ ≡ ℓ₁}.}
\begin{theorem}[Noninterference for Faceted Execution]
  Given:
\begin{agda}
o : (m : MultefInterface) → Fac m Bool → Fac m Bool
\end{agda}
We know that for all \ai{f₀, f₁ : Fac sec Bool},\\
 if \ai{f₀ ∼⟨ ℓ* ⟩ f₁} then o \ai{sec f₀ ∼⟨ ℓ* ⟩ o sec f₁}
\end{theorem}
\begin{proof}
  From \cref{thm:abstraction}, we know that the following holds:
  \begin{agda}
  ⟦o⟧ : (m₀ m₁ : MultefInterface)
     → (mᵣ : ⟦MultefInterface⟧ m₀ m₁)
     → (f₀ : Fac m₀ Bool) → (f₁ : Fac m₁ Bool)
     → (fᵣ : Facᵣ mᵣ Bool Bool ⟦Bool⟧ f₀ f₁)
     → Facᵣ mᵣ Bool Bool ⟦Bool⟧ (o m₀ f₀) (o m₁ f₁)
  \end{agda}
  We have chosen \ai{secᵣ : ⟦MultefInterface⟧ sec sec} so that:
  \begin{agda}
    Facᵣ secᵣ Bool Bool ⟦Bool⟧ f₀ f₁ = f₀ ∼⟨ ℓ* ⟩ f₁
  \end{agda}
  This means that:
  \begin{agda}
  ⟦o⟧ sec sec f₀ f₁ assume : o sec f₀ ∼⟨ ℓ* ⟩ o sec f₁
\end{agda}
  And is thus a valid Agda proof term for our theorem.
\end{proof}

\subsection{Noninterference for Core LIO}
\label{ni-lio}

Next we turn our attention to LIO.
We first explain our Agda port of LIO and then how
we use parametricity to prove noninterference for it.
Our port covers both the original LIO library \cite{LIO} and a more recent
extension to allow recovering from user-defined and IFC exceptions
\cite{StefanMMR17}, but leaves out state
and clearance
(the latter a concern not directly related to IFC that could also be easily added).
The interface of the core library can be seen in \cref{fig:lio}, but most of the
implementation has been elided for space.
Instead on focusing on the code, we give the intuition behind this library.

\begin{figure}
  \begin{agda}
record LIOInterface : Set₁ where
  field
    -- Labeled Values
    Labeled   : Set → Set
    label     : {A : Set} → Label → A → Labeled A
    labelOf   : {A : Set} → Labeled A → Label
    -- LIO Computations
    LIO       : Set → Set
    return    : {A : Set} → A → LIO A
    bind      : {A B : Set} → LIO A
             → (A → LIO B) → LIO B
    unlabel   : {A : Set} → Labeled A → LIO A
    toLabeled : {A : Set} → Label
             → LIO A → LIO (Labeled A)
    -- Exceptions
    throw     : {A : Set} → UserError → LIO A
    catch     : {A : Set} → LIO A
             → (E → LIO A) → LIO A

sec : LIOInterface
sec = { Labeled A = (A ⊎ E) × Label
      ; LIO A = (ℓc : Label) → Σ ((A ⊎ E) × Label)
                                  (λ r → ℓc ⊑ π₂ r)
      ; … }
  \end{agda}
  \caption{\label{fig:lio} Our Agda port of the core LIO library}
\end{figure}


The interface first defines \ai{Labeled A}, the abstract type of
labeled values of type \ai{A}.
If \ai{v} is a value of type \ai{A} and \ai{ℓ} is an IFC label, then we can use
the \ai{label ℓ v} operation to classify value \ai{v} at level \ai{ℓ}, which
results in a labeled value of type \ai{Labeled A}.
Once we have a \ai{Labeled} value we can use \ai{labelOf} to obtain its label,
which witnesses the fact that, in LIO, the label on data is public information.
Unlike the label, the value of a \ai{lv : Labeled A} is not public,
but protected precisely by \ai{labelOf lv}.
This means that it would not be secure to simply extract the value of \ai{lv},
so the \ai{LIO} interface provides no such operation.
Instead, to work with labeled values in a way that ensures IFC
we need to turn to \emph{labeled \ai{LIO} computations}.

An \ai{LIO} computation keeps track of a {\em current label},
which is the upper bound of all labeled values already inspected by the
computation.
\ai{LIO} threads through the current label, and in our case we also produce an
explicit proof that the current label can only increase in the IFC lattice, or
stay the same.
This (monotonic) state passing makes \ai{LIO} a monad, with \ai{return} and
\ai{bind} operations having analogous type signatures to those for faceted execution.
In addition, the \ai{LIO} monad provides an \ai{unlabel} operation that soundly returns the
value inside a labeled value \ai{lv} by increasing the current label by \ai{labelOf lv}:
\begin{agda}
  unlabel : {A : Set} → Labeled A → LIO A
\end{agda}
This allows \ai{LIO} computations to process labeled data, while using the
current label to track both explicit and implicit information flows (\IE flows
through the control flow of the program~\cite{sabelfeld2003language}).
Once we are done computing based on labeled values we can label the result and
restore the current label to what it was at the beginning of the current computation.
To prevent leaking information via the label of the final result, this label has
to be chosen in advance before inspecting any labeled data.
This functionality is implemented by the following operation:
\begin{agda}
  toLabeled : {A : Set} → Label
           → LIO A → LIO (Labeled A)
\end{agda}
The expression \ai{toLabeled ℓ lio} runs the \ai{lio} computation and if at the
end the current label is below \ai{ℓ} the result is labeled \ai{ℓ} and the
current label restored.
On the other hand, if at the end the current label is not below \ai{ℓ} we have
to signal an IFC error.
The original LIO~\cite{LIO} treated such errors as fatal and stopped execution,
however a more recent extension of LIO \cite{StefanMMR17} makes IFC errors
recoverable.
In the case of a wrongly annotated \ai{toLabeled} though, throwing an exception
would not be sound: we can restore the current label at the end only if that is
a control-flow join point.
To preserve this property, LIO returns instead a {\em delayed
exception}~\cite{hritcu2013all}, which is another kind of labeled value, labeled
with the originally chosen level \ai{ℓ}.
When unlabeling it the delayed exception is re-thrown, which is sound because, as
explained above, unlabeling a value raises the current label.

In addition to re-throwing delayed exceptions on \ai{unlabel}, LIO provides
standard primitives to \ai{throw} user exceptions and to \ai{catch} arbitrary
ones.
In order to achieve soundness, though, LIO also has to delay any such exceptions
at the end of \ai{toLabeled}.
To make debugging easier, all exceptions carry information such as the current
label at the time the exception was originally thrown together with a stack
trace. 
In addition, IFC exceptions record additional information about the involved
labels, when this information can be securely revealed (\EG when the label check
fails for \ai{toLabeled ℓ} it is secure to reveal the label \ai{ℓ}, but not the
current label that is not below \ai{ℓ}).

With this intuition in place, we can look at the actual definitions of
the \ai{Labeled} and \ai{LIO} types.
%
The definition of \ai{Labeled} is straightforward: a labeled
value is a pair of a \ai{Label} and either a result of type \ai{A}, or
a delayed exception of type \ai{E}, \ai{Labeled A = (A ⊎ E) × Label},
where \ai{⊎} denotes the tagged sum, or ``disjunctive union'':
\begin{agda}
  data _⊎_ : Set → Set → Set where
    inj₁ : {A B : Set} → A → A ⊎ B
    inj₂ : {A B : Set} → B → A ⊎ B
\end{agda}

The definition of \ai{LIO} meanwhile, is more involved.
A term of type \ai{LIO A} is a function that takes a current
label \ai{ℓc : Label} and produces a result that we
call a \emph{configuration}.
A configuration is an output label together with either a result of type \ai{A}
or a delayed exception of type \ai{E} (that's either a user exception or an IFC
exception whose details are omitted here).
However, as noted above we make the LIO computation also produce a proof that
the output label is at least as restrictive as the input label \ai{ℓc}.
This addition is opaque to the programmer, who programs against the abstract
interface of LIO, but is useful for simplifying our proofs, since it prevents
the logical relation below from being cluttered with this monotonicity proof.
Consequently, the definition of LIO is as follows:
\begin{agda}
-- Configuration
Cfg A = (A ⊎ E) × Label
-- LIO computation
LIO A = (ℓc : Label) -- Input label
      → Σ (Cfg A) -- Result
       (λ r → ℓc ⊑ π₂ r) -- Useful proof
\end{agda}

This definition uses a generalized sum type, or Σ-type,
to connect the proof of monotonic labels to the computation.
The Σ type-former is given by the following record type:
\begin{agda}
record Σ (A : Set) (B : A → Set) : Set where
  field  π₁ : A
         π₂ : B π₁
\end{agda}
and we additionally have the notation \ai{(a , b)} for \ai{record {π₁=a, π₂=b}}.
Finally, to avoid confusion we note that the \ai{_,_}-syntax is shared between
Σ-types and simple product types (\ai{A × B}), indeed the latter is an
instantiation of the former: \ai{A × B = Σ A (λ _ → B)}.

With this background in place, we turn to stating noninterference for \ai{LIO}.
As in the previous subsection, we could do this for any client function that is
parametric in the LIO interface, takes a labeled boolean input, and performs a
LIO computation returning a boolean as result:
\begin{agda}
o' : (m:LIOInterface) → Labeled m Bool → LIO m Bool
\end{agda}
While the soundness of IFC libraries such as LIO was sometimes only proved with
respect to such observers~\cite{LIO}, such statements are too specialized to
provide a sufficiently useful reasoning principle for most clients of the library.
Instead, we generalize ``\ai{Labeled m Bool}'' and ``\ai{LIO m Bool}'' above
to a first-order subset of Agda types, which includes base types, LIO-specific types and sum and products.
We define this subset syntactically, as a new inductive type in Agda:
\begin{agda}
data Univ : Set where
   bool    : Univ
   nat     : Univ
   error   : Univ
   labeled : Univ → Univ
   lio     : Univ → Univ
   _plus_  : Univ → Univ → Univ
   _times_ : Univ → Univ → Univ
\end{agda}
We interpret such syntax as the corresponding Agda type:
\begin{agda}
El : LIOInterface → Univ → Set
El m bool          = Bool
El m nat           = ℕ
El m error         = E
El m (labeled u)   = Labeled m (El m u)
El m (lio u)       = LIO m (El m u)
El m (u₀ plus u₁)  = El m u₀ ⊎ El m u₁
El m (u₀ times u₁) = El m u₀ × El m u₁
\end{agda}
Together, the above two constructions form a so-called ``universe'' of types --- hence the name \ai{Univ} for the inductive set.
This allows us to define a more general type of observers as functions taking
first-order values as inputs (``\ai{El m uᵢ}'') and returning first-order values (``\ai{El m uₒ}''):
\begin{agda}
o : (uᵢ uₒ : Univ)
    → (m:LIOInterface) → El m uᵢ → El m uₒ
\end{agda}

To state noninterference for \ai{o sec} we define \ai{ℓ*}-equivalence
by induction on the structure of our universe of first-order types.
For \ai{Bool}, \ai{ℕ}, and \ai{E} we define \ai{ℓ*}-equivalence simply as equality.
For pair types \ai{a × b} and sum types \ai{a ⊎ b} we define \ai{ℓ*}-equivalence pointwise.
In other words, two pairs \ai{(a₀ , b₀)} and \ai{(a₁ , b₁)} are \ai{ℓ*}-equivalent if \ai{a₀ ∼⟨ ℓ* ⟩ a₁}
and \ai{b₀ ∼⟨ ℓ* ⟩ b₁} and two sums \ai{injᵢ x} and \ai{injⱼ y} are \ai{ℓ*}-equivalent if
\ai{i ≡ j} and \ai{x ∼⟨ ℓ* ⟩ y}.
We say that two labeled values \ai{lv₀, lv₁ : Labeled sec a} are indistinguishable
at an observer label \ai{ℓ*}, written \ai{lv₀ ∼⟨ ℓ* ⟩ lv₁},\footnote{Formally
  speaking, our \ai{ℓ*}-equivalence definition in Agda has an extra universe
  argument, which we elide for readability here, but which disambiguates between
  labeled values and regular pairs in this definition.}
 if and only if:
\begin{enumerate}
  \item \ai{labelOf lv₀ ≡ labelOf lv₁} and
  \item if \ai{labelOf lv₀ ⊑ ℓ*} then \ai{π₁ lv₀ ∼⟨ ℓ* ⟩ π₁ lv₁}
\end{enumerate}

Point 1 says that the labels of \ai{lv₀} and \ai{lv₁} 
may never diverge from each other, since they are public information.
Point 2 says that if \ai{lv₀} (and therefore also \ai{lv₁}) is a public level,
then the payloads of \ai{lv₀} and \ai{lv₁} have to be \ai{ℓ*}-equivalent.
The recursive call into \ai{ℓ*}-equivalence happens at the sum type \ai{a ⊎ E},
which ensures that the payloads are either \ai{ℓ*}-equivalent values or equal errors.

Next we turn our attention to the relation for LIO computations.
Recall that \ai{c₀} and \ai{c₁} are internally state-passing computations where
the state is the current label: \ai{Label → Σ ((a + E) × Label) …}.
Because the current label is not necessarily public, but instead it protects
itself, the standard way to define \ai{ℓ*}-equivalence for the current label is
the following:
\begin{agda}
ℓc₀ ∼⟨ ℓ* ⟩ ℓc₁ = (ℓc₀ ⊑ ℓ* ∨ ℓc₁ ⊑ ℓ*) → ℓc₀ ≡ ℓc₁
\end{agda}
Two current labels are \ai{ℓ*}-equivalent if they are equal whenever one of them
is observable at level \ai{ℓ*}.
We can leverage this definition to define \ai{ℓ*}-equivalence for final
configurations of type \ai{(a + E) × Label}:\footnote{For readability
we are overloading notation here:
even though 
products are within our universe,
configurations get their \textit{ad-hoc} \ai{ℓ*}-equivalence definition.}
\begin{agda}
(r₀ , ℓc₀) ∼⟨ ℓ* ⟩ (r₁ , ℓc₁) =
    (ℓc₀ ∼⟨ ℓ* ⟩ ℓc₁) ∧
    (ℓc₀ ⊑ ℓ* ∧ ℓc₁ ⊑ ℓ* → r₀ ∼⟨ ℓ* ⟩ r₁)
\end{agda}
That is, for two configurations to be \ai{ℓ*}-equivalent, we require that the current
labels are \ai{ℓ*}-equivalent and if they are public then the results of the
computation should also be \ai{ℓ*}-equivalent.
With the above extensions of \ai{ℓ*}-equivalence in place, it is easy to define \ai{ℓ*}-equivalence for \ai{LIO} computations:
\begin{agda}
c₀ ∼⟨ ℓ* ⟩ c₁ = (ℓc₀ ℓc₁:Label) → ℓc₀ ∼⟨ ℓ* ⟩ ℓc₁ →
                 π₁ (c₀ ℓc₀) ∼⟨ ℓ* ⟩ π₁ (c₁ ℓc₁)
\end{agda}
In words, the above states that for any \ai{ℓ*}-equivalent initial current labels we obtain
\ai{ℓ*}-equivalent final configurations (the π₁ projections are needed to ignore
the proof part of the \ai{LIO} type).
This completes the definition of \ai{ℓ*}-equivalence for our universe,
and we can now state our noninterference theorem:

\begin{theorem}[Noninterference]\label{thm:lio-ni}
  Given:
\begin{agda}
o : (uᵢ uₒ : Univ)
    → (m:LIOInterface) → El m uᵢ → El m uₒ
\end{agda}
  For all \ai{v₀, v₁ : El sec uᵢ}, if
  %
  \begin{agda}
  v₀ ∼⟨ ℓ* ⟩ v₁
  \end{agda}
  %
  then
  \begin{agda}
   o sec v₀ ∼⟨ ℓ* ⟩ o sec v₁
 \end{agda}
\end{theorem}

The general strategy of the proof is the same as for Multef.
While out Agda proofs covers all the cases for \ai{uᵢ} and \ai{uₒ},
for simplicity, here we show only the special case
where ``\ai{El m uᵢ = Labeled m Bool}'' and ``\ai{El m uₒ = LIO m Bool}''.
In particular we use \cref{thm:abstraction} to obtain:
  \begin{agda}
⟦o⟧ : (m₀ m₁ : LIOInterface)
   → (mᵣ : ⟦LIOInterface⟧ m₀ m₁)
   → (l₀ : Labeled m₀ Bool) → (l₁ : Labeled m₁ Bool)
   → (lᵣ : Labeledᵣ mᵣ Bool Bool ⟦Bool⟧ l₀ l₁)
   → LIOᵣ mᵣ Bool Bool ⟦Bool⟧ (o m₀ l₀) (o m₁ l₁)
  \end{agda}
To use this result we first need to pick two relations
\ai{Labeledᵣ : ⟦Set → Set⟧ Labeled Labeled} and
\ai{LIOᵣ : ⟦Set → Set⟧ LIO LIO} and prove that
relatedness at these relations is respected by the LIO operations.
Moreover, to be useful for proving noninterference these relations have to
specialize
to the \ai{ℓ*}-equivalence relations used in the noninterference statement above,  for all the types in our \ai{Univ} universe (\EG for \ai{Labeled Bool}).

For example, to define \ai{Labeledᵣ}, we generalize the definition of
\ai{ℓ*}-equivalence to ensure that when \ai{Aᵣ} coincides with
\ai{ℓ*}-equivalence, then \ai{Labeledᵣ A₀ A₁ Aᵣ} does too:
\begin{agda}
Labeledᵣ A₀ A₁ Aᵣ lv₀ lv₁ =
 (labelOf m₀ lv₀ ≡ labelOf m₁ lv₁) ∧
 (labelOf m₀ lv₀ ⊑ ℓ* → (Aᵣ ⟦⊎⟧ Eᵣ) (π₁ lv₀) (π₁ lv₁))
\end{agda}
(Where the relation \ai{Aᵣ ⟦⊎⟧ Bᵣ} relates two values if they are
either \ai{inj₁} and related by \ai{Aᵣ} or \ai{inj₂} and related by \ai{Bᵣ}.)
In the second conjunct, we require that if the level of the two labeled values
is public then either they are both errors related at \ai{Eᵣ}, or they both
carry values of type \ai{A} that are related at \ai{Aᵣ}.
The relation \ai{Eᵣ : E → E → Set} relates two delayed exceptions if and only if
they are the same.{}
We do a similar generalization types in the \ai{Univ} universe and \ai{ℓ*}-equivalence
to arbitrary types for defining \ai{LIOᵣ}.

The main part of our Agda proof of \cref{thm:lio-ni} is showing that the LIO
operations respect the \ai{Labeledᵣ} and \ai{LIOᵣ} relations.
%
Some of the LIO operations have straightforward proofs (\ai{label},
\ai{return}, \ai{labelOf}, \ai{unlabel}, and \ai{throw}),
while the higher-order operations (\ai{bind}, \ai{toLabeled}, and \ai{catch}),
have slightly more interesting proofs that rely on the monotonicity of the
current label.
Fortunately, all these proofs are pleasantly short adding up to around 360 lines
of Agda for the complete noninterference proof in our supplementary material.
%


The short and fully mechanized Agda proof that we describe above can be contrasted
with the previous 
partially mechanized proof for LIO~\cite{StefanMMR17}.
This previous proof shows noninterference for an abstract calculus
\emph{without} exception handling and state, covering a strict subset of the
features of the library implementation that we have verified here.
Their proof technique, to show a simulation between evaluation of an LIO term
with secrets and the same term with the secrets erased, is standard but
cumbersome.
Consequently, their proof amounts to over 3000 lines of Coq, even if it is not
fully mechanized and it only covers a small calculus, not a library
implementation in terms of type abstraction.
Their proof could probably be finished and made shorter by using more tactic
automation or a better proof strategy~\cite{hritcu2013all}, yet it seems hard
to match the conceptual simplicity and compactness of our parametricity-based proof.


\section{Transparency}
\label{sec:transparency}
One of the primary justifications for faceted semantics is the
so-called \emph{transparency} theorem \cite{MF}.
In short, transparency states that: for any program $p$ that is noninterfering
under a non-faceted, ``standard'' semantics, the behavior of $p$ is
preserved when $p$ is run with faceted semantics.
Intuitively, this means that there are no false alarms with faceted execution:
if the program is noninterfering to begin with, facets don't change anything.
This is unlike systems like LIO, where false alarms are a problem that the
programmer has to work around by adhering to proper programming style.

In our setting, this transparency property can be reformulated in terms of one
of the key lemmas used to prove both noninterference and transparency for
traditional faceted calculi: \emph{faceted evaluation simulates standard
evaluation}~\cite{MF, OptimisingFSME}.

To make sense of what this means in our context we need to explain the distinction
between faceted and standard evaluation.
Luckily, it is straightforward for us to define what we mean by different
semantics for the same program: we simply give different implementations of the
\ai{MultefInterface}!
In particular, the faceted semantics was already defined as the
\ai{MultefImplementation} module in \cref{fig:multef} from \cref{DIFC-as-a-library}.

\begin{figure}
\begin{agda}
Maybe : Set → Set
Maybe A = A ⊎ ⊤

just : {A : Set} → A → Maybe A
just a = inj₁ a

nothing : {A : Set} → Maybe A
nothing = inj₂ tt

⟦Maybe⟧ : (A₀ A₁ : Set) → (Aᵣ : A₀ → A₁ → Set)
       → Maybe A₀ → Maybe A₁ → Set
⟦Maybe⟧ _ _ Aᵣ = Aᵣ ⟦⊎⟧ ⟦⊤⟧
\end{agda}
\squeezeup
  \caption{\label{fig:maybe} The \ai{Maybe} type former and its relational interpretation.}
\end{figure}

To define the standard semantics, we first introduce the
\ai{Maybe} (also known as ``option'') type former, as a special case
of the \ai{⊎} type, in \cref{fig:maybe}.
With this in place, we define the standard semantics of Multef in
\cref{fig:multef-std}.
Most of the definitions are unsurprising: \ai{Fac A} is \ai{Maybe A}, while
\ai{return} and \ai{bind} are standard for the \ai{Maybe} monad.
The only potentially surprising definition is \ai{facet ℓ f₀ f₁ = nothing}.
To understand it, note that the standard phrasing of the transparency
property in the literature makes reference to \emph{facet-free} programs
\cite{FSME, OptimisingFSME, MF}.
In our setting, we cannot easily talk about such ``facet-free'' programs, because all
programs we study are clients of the \ai{MultefInterface},
%
so instead we make do by talking about programs that do not return \ai{nothing}
under evaluation in the standard semantics.

\begin{figure}
  \begin{agda}
module Std where
  Fac : Set → Set
  Fac a = Maybe a

  facet : {A : Set} → L → Fac A
       → Fac A → Fac A
  facet ℓ f₀ f₁ = nothing

  return : {A : Set} → A → Fac A
  return a = just a

  bind : {A B : Set} → Fac A
      → (A → Fac B) → Fac B
  bind f c = case f of \
    { inj₁ a  → c a
    ; inj₂ tt → inj₂ tt }

std : MultefInterface
std = record {Std}
  \end{agda}
\squeezeup
  \caption{\label{fig:multef-std} Standard Semantics of Multef}
\end{figure}

\begin{theorem}[Transparency]
  \label{thm:transparency}
  Fix a label \ai{ℓ* : Label}.
  Given any \ai{b : Bool}, define the faceted value \ai{fᵇ} as having value
  \ai{b} for observers that can see data labeled \ai{ℓ*}, and value \ai{false}
  otherwise as:
  \begin{agda}
fᵇ = facet sec Bool ℓ* (return sec b)
                       (return sec false)
  \end{agda}
  For any client function \ai{o}, parametric in \ai{MultefInterface}:
  \begin{agda}
o : (m : MultefInterface) → Fac m Bool → Fac m Bool,
  \end{agda}
  which does not crash under the standard semantics (\ai{std}) when given the
  non-faceted constant \ai{b} as input:
  \begin{agda}
o std (just b) ̸≡ nothing,
  \end{agda}
  then running \ai{o std (just b)} yields the same result from the point of
  view of an observer at level \ai{ℓ*} as running \ai{o} with the faceted
  semantics (\ai{sec}) on input \ai{fᵇ}:
  \begin{agda}
   o std (just b) ≡ just Bool (o sec fᵇ ℓ*)
\end{agda}
\end{theorem}

To prove \cref{thm:transparency} we need to relate the execution of
\ai{o std} with the execution of \ai{o sec}.
According to our setup, \cref{thm:abstraction} gives us that if
\ai{o : (m : MultefInterface) → T} for some type \ai{T}, then:
\begin{agda}
⟦o⟧ : (m₀ m₁ : MultefInterface)
   → (mᵣ : ⟦MultefInterface⟧ m₀ m₁)
   → ⟦T⟧ (o m₀) (o m₁)
\end{agda}
In particular, if we can provide some \ai{std-secᵣ : ⟦MultefInterface⟧ std sec},
then parametricity lets us relate \ai{o std} and \ai{o sec}.

As we have seen with the previous proofs in this paper, the key to getting this
to work is picking the correct instantiation of \ai{Facᵣ} in \ai{std-secᵣ}.
In this case the choice is clear from the theorem that we are trying to prove.
We want that, when the standard (non-faceted) result is not \ai{nothing}, the
projection at \ai{ℓ*} of the value resulting from the faceted
execution is related to the standard value.
In other words, we pick the following definition of \ai{Facᵣ} in \ai{std-secᵣ}:
\begin{agda}
Facᵣ std-secᵣ = λ A₀ A₁ Aᵣ f₀ f₁ → f₀ ̸≡ nothing →
  ⟦Maybe⟧ A₀ A₁ Aᵣ f₀ (just (project f₁ ℓ*))
\end{agda}
From this definition, filling out \ai{facetᵣ}, \ai{returnᵣ}, and \ai{bindᵣ} is straightforward.
With the definition of \ai{std-secᵣ} in place, we can tackle the proof of \cref{thm:transparency}.

\begin{proof}[Proof of \cref{thm:transparency}]
  From:
  \begin{agda}
o : (m : MultefInterface) → Fac m Bool → Fac m Bool
  \end{agda}
  We obtain by parametricity:
\begin{agda}
⟦o⟧ : (m₀ m₁ : MultefInterface)
   → (mᵣ : ⟦MultefInterface⟧ m₀ m₁)
   → (f₀ : Fac m₀ Bool) → (f₁ : Fac m₁ Bool)
   → (fᵣ : Facᵣ mᵣ Bool Bool ⟦Bool⟧ f₀ f₁)
   → Facᵣ mᵣ Bool Bool ⟦Bool⟧ (o m₀ f₀) (o m₁ f₁)
\end{agda}
  We pick \ai{m₀ = std}, \ai{m₁ = sec}, \ai{mᵣ = std-secᵣ}, \ai{f₀ = just b}, \ai{f₁ = fᵇ},
  and let \ai{fᵣ} be some (omitted) proof that \ai{f₀} and \ai{f₁} are appropriately
  related\footnote{The reader will find its definition in the Agda mechanization of this paper}
\begin{agda}
fᵣ : Facᵣ std-secᵣ Bool Bool ⟦Bool⟧ f₀ f₁,
\end{agda}
whose type becomes this after unfolding definitions:
\begin{agda}
fᵣ : just b ̸≡ nothing →
       ⟦Maybe⟧ Bool Bool _≡_ (just b)
                             (just (project fᵇ ℓ*)).
\end{agda}
  This gives us the following instantiation for \ai{⟦o⟧}:
  \begin{agda}
⟦o⟧ std sec std-secᵣ (just b) fᵇ fᵣ :
  o std (just b) ̸≡ nothing →
    ⟦Maybe⟧ Bool Bool _≡_ (o std (just b))
                       (just (project (o sec fᵇ) ℓ*)).
  \end{agda}
  This is sufficient to easily establish the theorem.
\end{proof}

The key takeaway from this proof is the generality of parametricity as a proof
technique.
While previous proofs have focused on the connection between noninterference
and parametricity
\cite{SimpleNoninterferenceFromParametricity,bowman2015noninterference,ngo2019typed},
the proof above shows that parametricity can also be useful for proving other
interesting meta-theoretical properties about security libraries.



\section{Related and Future Work}

Dynamic IFC libraries, like LIO \cite{LIO} and Multef~\cite{FSME},
promise to provide noninterference guarantees without the need for a specialized tool-chain.
Embedding such libraries in existing languages allows programmers to reuse the
functionality and library ecosystem of the host language.
Case studies show that this is a promising direction for IFC~\cite{HAILS, LIO, parker2019lweb}.

\citet{StefanMMR17} provide a noninterference proof for LIO partially mechanized
in Coq, over which we improve in several ways.
First, we cover a larger subset of LIO, by additionally supporting recoverable
exceptions~\cite{StefanMMR17, hritcu2013all}, thrown by clients of by the IFC
mechanism itself.
Second, our complete proof fits in just 360 lines of Agda, while the proof of
noninterference of \citet{StefanMMR17} is more than 3000 lines of Coq.
We achieve this order-of-magnitude conciseness improvement by using
logical-relations to reason about LIO and by relying on \cref{thm:abstraction}
to automatically derive the relational interpretation of the ``standard'' parts
of the language (crucially $\lambda$-abstraction): a large part of this proof
one would traditionally need to carry out manually.
%

Third, and importantly, we certify \emph{an implementation of the library in
terms of type abstraction}, rather than a model of the library that ignores
it--an issue which affects much of the library-based IFC literature.
This issue is not surprising, because in order to be practical the libraries are
implemented in languages such as Haskell which are lacking a formal semantics
based on logical relations.
%
Formalizing library-based IFC would require developing such a semantics first
--- a daunting task for a full-featured language.
In contrast, we use a less mainstream language, Agda, but, in return, there is
no informal, unverified, modeling step remaining in our approach.

To be sure, this means that we also do not claim to provide guarantees
for the \emph{Haskell} implementations of the libraries, but rather the
Agda implementations in this paper.
Specifically, this means that we only protect against Agda attackers, not the
potentially more powerful Haskell attackers that previous approaches claim to
defend against.
However, we believe this paper demonstrates that as techniques for reasoning
about parametricity for languages like Haskell are developed, the necessary
meta-theory for practical IFC libraries will come within reach of our
techniques.
%
%
\iffull
\ch{Worth mentioning in passing that Haskell is also getting ``real''
  dependent types?}
\fi

Developing tools to cover the full Haskell language requires some innovation
over the Agda parametricity theorems we rely on in this paper.
Specifically, Haskell has a number of features not covered by the existing
theory for Agda, including IO, concurrency, lazyness, and mutable state.
Correctly reasoning about these features may require using a more expressive
ambient logic than what is afforded by having just a dependently typed
language like Agda.
For example, separation logic
has been used in
previous work to provide parametricity-like reasoning for a subset of the
langauge features needed to cover the full Haskell language
\cite{10.1145/3341709}.


Another difference between Haskell and Agda is that Agda is a total language, so
all code written in it is guaranteed to terminate and so the termination channel
\cite{hedin2012perspective} is a non-issue.
Extending our techniques to deal with non-termination, and specifically to deal
with libraries that protect against termination leaks (\EG
\cite{stefan2012addressing,FSME}), is interesting future work.
In particular, we hope to build on various monadic representations of
non-termination in dependent type theory \cite{McBride15, XiaZHHMPZ20}.



Last but not least, as already mentioned in the introduction,
\citet{SimpleNoninterferenceFromParametricity} introduced the technique of using
parametricity for dependent types to show noninterference for a security library.
They are concerned with \emph{static} IFC libraries, while we show that the
technique extends naturally to libraries for \emph{dynamic} IFC libraries.
Furthermore, in \cref{thm:transparency} we show that the parametricity
proof-technique works well for proving meta-theoretical properties other
than noninterference for security libraries using this approach.

 

\section{Conclusion}
This paper shows the versatility of parametricity as a proof technique for
language-based security.
Specifically, we show that parametricity can be used to prove noninterference
for two different dynamic IFC libraries as well as transparency for the faceted values one.
Ours are the first proofs of noninterference for the implementation of dynamic
IFC libraries in terms of type abstraction.
%
%
Furthermore, using parametricity allows us to give compact yet fully
machine-checked proofs.

We believe that the simplicity of the proof techniques used in this paper will
be key to scaling noninterference proofs to cover the actual code of
feature-rich IFC libraries, including for instance mutable state \cite{BuirasSR14, StefanMMR17},
concurrency \cite{stefan2012addressing, buiras2013lazy, buiras2013library},
parallelism \cite{10.1007/978-3-030-17138-4_1}, and cryptography \cite{CLIO}.

\section*{Acknowledgements}

We are grateful to the CSF 2021 reviewers for their helpful feedback.
Algehed was partially supported by the Wallenberg AI, Autonomous
Systems and Software Program (WASP) funded by the Knut and Alice Wallenberg
Foundation and The Osher Endowment through project Ghost: Exploring the
Limits of Invisible Security.
Bernardy is partially supported by the Swedish Research Council Grant
No.~2014-39, which funds the Centre for Linguistic Theory and Studies
in Probability (CLASP) in the Department of Philosophy, Linguistics,
and Theory of Science at the University of Gothenburg.
Hri\cb{t}cu was in part supported by the
\href{https://erc.europa.eu}{European Research Council} under
\href{https://secure-compilation.github.io}{ERC Starting Grant SECOMP} (715753).

\bibliographystyle{IEEEtranN}
\bibliography{main}

\end{document}